\pdfoutput=1

\documentclass[sigconf]{acmart}

\DeclareOldFontCommand{\bi}{\bfseries\itshape}{\bfseries\itshape}
\makeatletter
\@ifundefined{tabnote}{
  
  }{}

\def\aptLtoXcmd#1#2{#2}
\long\def\aptLtoX{\@ifnextchar[{\@aptLtoX}{\@aptLtoX[]}}
\long\def\@aptLtoX[#1]#2#3{#3}
\def\aptLtoXdel#1{}
\makeatother

\newcommand{\bigO}[1]{\mathchoice{O\left(#1\right)}{O(#1)}{O(#1)}{O(#1)}} 
 
\DeclareMathOperator{\loglog}{loglog} 
\newcommand{\expmm}{\omega} 
\newcommand{\timepm}[1]{\mathchoice{\operatorname{\mathsf{M}}\!\left(#1\right)}{\operatorname{\mathsf{M}}(#1)}{\operatorname{\mathsf{M}}(#1)}{\operatorname{\mathsf{M}}(#1)}} 
\newcommand{\timepmm}[2]{\mathchoice{\operatorname{\mathsf{MM}}\!\left(#1,#2\right)}{\operatorname{\mathsf{MM}}(#1,#2)}{\operatorname{\mathsf{MM}}(#1,#2)}{\operatorname{\mathsf{MM}}(#1,#2)}} 
\newcommand{\timepmmrec}[2]{\mathchoice{\operatorname{\overline{\mathsf{MM}}}\!\left(#1,#2\right)}{\operatorname{\overline{\mathsf{MM}}}(#1,#2)}{\operatorname{\overline{\mathsf{MM}}}(#1,#2)}{\operatorname{\overline{\mathsf{MM}}}(#1,#2)}} 

\newcommand{\ZZ}{\mathbb{Z}} 
\newcommand{\NN}{\mathbb{N}} 
\newcommand{\ZZp}{\mathbb{Z}_{> 0}} 
\newcommand{\field}{\mathbb{F}} 
\newcommand{\pring}{\field[x]} 
\newcommand{\vecspace}[1]{\field^{#1}} 
\newcommand{\matspace}[2]{\field^{#1 \times #2}} 
\newcommand{\pvecspace}[1]{\pring^{#1}}
\newcommand{\pmatspace}[2]{\pring^{#1 \times #2}} 
\newcommand{\card}[1]{\# #1}  
\newcommand{\range}[1]{\{1,\ldots,#1\}}  

\newcommand{\mat}[1]{#1} 
\newcommand{\matcol}[2]{{#1}_{*,#2}} 
\newcommand{\matcols}[2]{{#1}_{*,#2}} 
\newcommand{\submat}[3]{{#1}_{#2,#3}} 

\newcommand{\trsp}[1]{#1^\mathsf{T}}

\newcommand{\ident}[1]{\mat{I}_{#1}} 
\newcommand{\matz}{\mat{0}}  
\newcommand{\rank}[1]{\operatorname{rank}(#1)}

\newcommand{\sumTuple}[1]{|#1|} 
\newcommand{\cdeg}[2][]{\mathrm{cdeg}_{{#1}}(#2)}
\newcommand{\shift}{s}

\newcommand{\coeffs}[2]{\operatorname{coeffs}(#1,#2)}

\def\rem{\, \text{rem}\, }
\newcommand{\kry}[3]{\operatorname{K}_{#3}(#1,#2)}
\newcommand{\orb}[2]{\operatorname{Orb}(#1,#2)}

\newcommand{\expoH}{{c_1}}
\newcommand{\expofinal}{c}
\newcommand{\myparagraph}[1]{\smallskip\emph{#1.}} 
\newcommand{\avgdeg}{t} 
\newcommand{\thres}{{t}} 
\newcommand{\DD}{{\ell}} 

\usepackage{algorithm}
\usepackage{algpseudocode}     

\aptLtoXcmd{}{\algrenewcommand\textproc{\texttt}}

\aptLtoXcmd{}{\algrenewcomment[1]{\hfill \(\triangleright\) \emph{\small #1}} 
\algnewcommand{\CommentLine}[1]{\(\triangleright\;\;\) \emph{\small #1} \(\;\;\triangleleft\)} 
\algnewcommand{\InlineIf}[2]{
   \algorithmicif\ #1\ \algorithmicthen\ #2}
\algnewcommand{\InlineFor}[2]{\algorithmicfor\ #1\ \algorithmicdo\ #2} 
\algrenewcommand\Call[2]{\nameref{#1}\ifthenelse{\equal{#2}{}}{}{\ensuremath{(#2)}}}}

\aptLtoXcmd{}{
}

\makeatletter
\newcommand{\StateX}[1]{%
  \setlength\@tempdima{\algorithmicindent}%
  \Statex\hskip\dimexpr#1\@tempdima\relax}
\makeatother
\makeatletter
\newcommand{\algoCaptionLabel}[2]{
     \caption[\textproc{#1}]{\textproc{#1}\ifthenelse{\equal{#2}{}}{}{$(#2)$} }%
     \NR@gettitle{\textproc{#1}}%
      \label{algo:#1}
     }%
\makeatother

\usepackage[capitalise]{cleveref}

\copyrightyear{2024}
\acmYear{2024}
\setcopyright{licensedothergov}\acmConference[ISSAC '24]{International
Symposium on Symbolic and Algebraic Computation}{July 16--19,
2024}{Raleigh, NC, USA}
\acmBooktitle{International Symposium on Symbolic and Algebraic
Computation (ISSAC '24), July 16--19, 2024, Raleigh, NC, USA}
\acmDOI{10.1145/3666000.3669715}
\acmISBN{979-8-4007-0696-7/24/07}

\begin{document}

\title{Computing Krylov iterates in the time of matrix multiplication}

\author{Vincent Neiger}
\orcid{0000-0002-8311-9490}
\affiliation{
\institution{Sorbonne Université, CNRS, LIP6}
\city{F-75005 Paris}
  \country{France}
}

\author{Cl\'ement Pernet}
\orcid{0000-0001-6970-0417}
\affiliation{
  \institution{Univ. Grenoble Alpes, Grenoble INP}
  \department{CNRS, LJK UMR 5224}
  %\streetaddress{700 avenue centrale, IMAG --- CS 40700}
  \city{Grenoble}
  \country{France}
}

\author{Gilles Villard}
\orcid{0000-0003-1936-7155}
\affiliation{
  \institution{CNRS, Univ. Lyon,  ENS de Lyon, Inria,}
  \department{UCBL, LIP UMR 5668}
  %\streetaddress{46, Allée d'Italie}
  \city{Lyon}
  \country{France}
}

\begin{CCSXML}
<ccs2012>
<concept>
<concept_id>10010147.10010148.10010149.10010150</concept_id>
<concept_desc>Computing methodologies~Algebraic algorithms</concept_desc>
<concept_significance>500</concept_significance>
</concept>
<concept>
<concept_id>10003752.10003809</concept_id>
<concept_desc>Theory of computation~Design and analysis of algorithms</concept_desc>
<concept_significance>500</concept_significance>
</concept>
</ccs2012>
\end{CCSXML}

 \keywords{Krylov iteration; Frobenius normal form; Polynomial linear algebra.}

\begin{abstract}
  Krylov methods rely on iterated matrix-vector products \(A^k u_j\) for an
  \(n\times n\) matrix \(A\) and vectors \(u_1,\ldots,u_m\). The space spanned
  by all iterates \(A^k u_j\) admits a particular basis --- the \emph{maximal
  Krylov basis} --- which consists of iterates of the first vector \(u_1, Au_1,
  A^2u_1,\ldots\), until reaching linear dependency, then iterating similarly
  the subsequent vectors until a basis is obtained. Finding minimal polynomials
  and Frobenius normal forms is closely related to computing maximal Krylov
  bases. The fastest way to produce these bases was, until this paper,
  Keller-Gehrig's 1985 algorithm whose complexity bound \(\bigO{n^\omega
  \log(n)}\) comes from repeated squarings of \(A\) and logarithmically many
  Gaussian eliminations. Here \(\omega>2\) is a feasible exponent for matrix
  multiplication over the base field.  We present an algorithm computing the
  maximal Krylov basis in \(\bigO{n^\omega\loglog(n)}\) field operations when
  \(m \in \bigO{n}\), and even \(\bigO{n^\omega}\) as soon as
  \(m\in\bigO{n/\log(n)^\expofinal}\) for some fixed real \(\expofinal>0\).  As
  a consequence, we show that the Frobenius normal form together with a
  transformation matrix can be computed deterministically in \(\bigO{n^\omega
  (\loglog(n))^2}\), and therefore matrix exponentiation~\(A^k\) can be performed
  in the latter complexity if \(\log(k) \in \bigO{n^{\omega-1-\varepsilon}}\)
  for some fixed $\varepsilon>0$.  A key idea for these improvements is to rely on fast
  algorithms for \(m\times m\) polynomial matrices of average degree \(n/m\),
  involving high-order lifting and minimal kernel bases.
\end{abstract}

\thanks{%
  The authors are supported by \grantsponsor{anr}{Agence nationale de la
  recher\-che (ANR)}{https://anr.fr} projects
  \grantnum{anr}{ANR-23-CE48-0003-01} \textsc{CREAM},
  \grantnum{anr}{ANR-22-PECY-0010-01} \textsc{PEPR Cryptanalyse} and
  \grantnum{anr}{ANR-21-CE39-0006} \textsc{SANGRIA};
  by the joint
  ANR-\grantsponsor{fwf}{Austrian Science Fund FWF}{https://www.fwf.ac.at/en/}
  projects \grantnum{anr}{ANR-22-CE91-0007}
  \textsc{EAGLES}
  and by the \grantsponsor{afosr}{EOARD-AFOSR}{} project
  \grantnum{afosr}{FA8665-20-1-7029}.
}

\maketitle

%%%%%%%%%%%%%%%%%%%%%%%%%%%%%%%%%%%%%%%%%%%%%%%%%%%%%%%%%%%%%%%%%%%%%%%%%%%%%%%%%%%
%
%              INTRODUCTION
%
%%%%%%%%%%%%%%%%%%%%%%%%%%%%%%%%%%%%%%%%%%%%%%%%%%%%%%%%%%%%%%%%%%%%%%%%%%%%%%%%%%%

\section{Introduction}
\label{sec:intro}

We present a new deterministic algorithm for the computation of some specific
Krylov matrices, which play a central role in determining the structure of
linear operators. To a matrix \(A\in \matspace{n}{n}\) over an arbitrary
commutative field \(\field\), a (column) vector \(u\in \vecspace{n}\) and a
nonnegative integer \(d \in \NN\), we associate the \emph{Krylov matrix}
\(\kry{A}{u}{d}\) formed by the first \(d\) iterates of \(u\) through \(A\):
\begin{equation}
  \label{eqn:krylov_iterate}
  \kry{A}{u}{d} =
  \begin{bmatrix}
    u & Au & \cdots & A^{d-1}u
  \end{bmatrix}
  \in \matspace{n}{d}.
\end{equation}
More generally, to \(m\) vectors \(U = [u_1 \;\; \cdots \;\; u_m]
\in\field^{n\times m}\) and a tuple \(d=(d_1,\ldots,d_m)\in \NN^m\), we
associate the Krylov matrix
\begin{equation}
  \label{eqn:krylov_matrix}
  \kry{A}{U}{d}
  =
  \begin{bmatrix}
    \kry{A}{u_1}{d_1} & \cdots & \kry{A}{u_m}{d_m}
  \end{bmatrix}
  \in \matspace{n}{\sumTuple{d}},
\end{equation}
where \(\sumTuple{d}\) is the sum \(d_1+\cdots+d_m\). (Note that it will prove
convenient to allow \(m=0\) and \(d_i=0\).) Such matrices are used to construct
special bases of the \(A\)-invariant subspace 
\[
  \orb{A}{U} = \operatorname{Span}_\field(
                  \{A^i u_j, i\in \NN, j\in \range{m}\}
                  ).
\]
Indeed, for a given \(A\) and \(U\) there always exists a tuple \(d\) such that
the columns of \(\kry{A}{U}{d}\) form a basis of \(\orb{A}{U}\)
(Section~\ref{sec:find_degrees:KG}).  In this paper, we focus on the computation of
the unique such basis of \(\orb{A}{U}\) whose tuple~\(d\) is the
lexicographically largest one~\cite[Sec.\,5]{KeGe85}. We call this \(d\) the
\emph{maximal (Krylov) indices} of \(\orb{A}{U}\), and the corresponding basis
the \emph{maximal Krylov basis}.

Our main result is a deterministic algorithm that computes the maximal Krylov
basis in \(\bigO{n^\expmm\loglog(n)}\) field operations when $m\in O(n)$, where
$\expmm >2$ is a feasible exponent for the cost of square matrix multiplication
over~$\field$~\cite{DWZ23,VWXXZ24,ADVWXXZ24}. The bound becomes $O(n^\expmm)$ as
soon as the number \(m\) of initial vectors in \(U\) is in
\(\bigO{n/\log(n)^\expofinal}\) for some constant \(\expofinal>0\)
(see \cref{thm:maxkrylovbasis}). This is an improvement over the best previously
known complexity bound \(\bigO{n^\expmm \log(n)}\), for an algorithm due to
Keller-Gehrig~\cite{KeGe85}. In particular, to get  down to $O(n^\expmm)$, we
avoid an ingredient that is central in the latter algorithm and related ones,
which is to compute logarithmically many powers of \(A\) by repeated squaring
(see Section~\ref{sec:find_degrees:KG}). \enlargethispage{-12pt}

\myparagraph{Overview of the approach}
The main idea is to use operations on polynomial matrices rather than linear
transformations. In this direction, we are following in the footsteps of e.g.\
\cite{ZLS15}, where polynomial matrix inversion is exploited to compute
sequences of matrix powers, and \cite{NeigerPernet21}, where polynomial matrix
normal forms and block-triangular decompositions allow the efficient
computation of the characteristic polynomial.
A key stage we introduce consists in transforming between left and right matrix
fraction descriptions:
\begin{equation}
    \label{eqn:series_expansion_intro}
    S(x)T(x)^{-1} = (I - x A)^{-1} U = \sum_{k\ge 0} x^k A^k U,
\end{equation}
with $S\in \pmatspace{n}{m}$ and $T\in \pmatspace{m}{m}$. For $m\leq n$, one
may see Eq.~(\ref{eqn:series_expansion_intro}) as considering a compressed
description \(ST^{-1}\), which has larger polynomial degrees
and smaller matrix dimensions than the description $(I - x A)^{-1} U$. The
power series expansion of \(ST^{-1}\), when suitably truncated, produces a
Krylov basis.  

After some preliminary reminders on polynomial matrices in Section~\ref{sec:prelimin},
the first algorithms are given in Section~\ref{sec:KrylovMatrix}. There, our
contribution is specifically adapted to the case where the compression is fully
effective, that is, when $m$ is away from \(n\) (at least slightly:
\(\bigO{n/\log(n)^\expofinal}\) suffices).
In this case, appropriate $S$ and
$T$ are computed using $O(n^\expmm)$ arithmetic operations
thanks to the kernel basis algorithm of
\cite{ZhouLabahnStorjohann12} and its analysis
in~\cite{JeannerodNeigerSchostVillard2017,NeigerPernet21}
(\cref{subsec:kernel}). 

The next steps are to determine the maximal indices and to compute a truncated
series expansion of \(ST^{-1}\).
First we explain in Section~\ref{sec:find_degrees} that the maximal indices of
\(\orb{A}{U}\) are obtained by working, equivalently to
Eq.~(\ref{eqn:series_expansion_intro}), from certain denominator matrices of
$(xI-A)^{-1}U$ (see \cref{eq:matfracrev}). The indices are computed as diagonal
degrees if the denominator is triangular (Lemma~\ref{lem:hermite_gives_degrees}).
It follows that a Hermite form computation allows us to obtain them
efficiently~\cite{LabahnNeigerZhou2017}.
We then give in Section~\ref{sec:compute_krylow_few} an algorithm for computing a
Krylov matrix \(\kry{A}{U}{d}\) for an arbitrary given $d$, which we will apply
afterwards with the maximal indices for \(d\).
According to Eq.~(\ref{eqn:series_expansion_intro}) and given a tuple $d =({{d_1,
\ldots, d_m}})\in \NN^m$, this Krylov matrix can be obtained from the expansion
of $ST^{-1}$ with column $j$ truncated modulo $x^{d_j}$ for $1\leq j \leq m$.
To deal with the unbalancedness of degrees and truncation order, this expansion
is essentially computed using high-order lifting~\cite{Storjohann2003},
combined with the partial linearization technique of~\cite{GuSaStVa12}.
(To avoid diverting the focus of Section~\ref{sec:KrylovMatrix}, details about this are
deferred until Section~\ref{sec:polmat}.) If \(m\in\bigO{n/\log(n)^\expofinal}\),
both the maximal indices and $\kry{A}{U}{d}$ for any~$d$ such that
\(\sumTuple{d}=O(n)\) can be computed using $O(n^\expmm)$ arithmetic
operations. 

Our general algorithm computing maximal Krylov bases is given in
Section~\ref{sec:KrylovBasis}. The ability to reduce the cost for certain $m$, as seen
above, allows to improve the general case $m=O(n)$ and obtain the complexity
bound \(\bigO{n^\expmm\loglog(n)}\). 
Algorithm \texttt{MaxKrylovBasis} is a hybrid one, using the Keller-Gehrig
strategy as well as polynomial matrices.
In the same spirit as the approach in~\cite[Sec.\,5]{KeGe85}, 
we start with the partial computation of a maximal Krylov basis,
from the whole $U$ but for only a few iterations,
i.e.\ only considering iterates \(A^{k_j} u_j\) for \(k_j\) bounded by $\log(n)^\expofinal$. 
This allows us to isolate \(\bigO{n/\log(n)^\expofinal}\) vectors 
for which further iterations are needed.
A maximal basis is then computed from these vectors, based on polynomial matrix operations.  
The final maximal basis is obtained by appropriately merging the short and long sequences of Krylov
iterates henceforth available, via fast Gaussian elimination~\cite[Sec.\,4]{KeGe85}.

\myparagraph{Frobenius normal form and extensions}
Krylov matrices are a fundamental tool for decomposing a vector space
with respect to a linear operator (see e.g.\ \cite{giesbrecht1995nearly,Sto00}
for detailed algorithmic treatments). 
The related Frobenius normal form, and the Kalman decomposition
 for linear dynamical
systems~\cite{kal63,Kailath80},
are briefly discussed in
Section~\ref{sec:decomp}. The fastest known deterministic algorithm for the Frobenius
form is found in~\cite[Prop.\,9.27]{Sto00}; the employed approach reveals
in particular that the problem reduces to computing
$O(\loglog(n))$ maximal Krylov bases. The cost for the Kalman decomposition is
bounded by that of the computation of a constant number of Krylov bases.
Thus our results improve the complexity bounds for these two problems. Matrix
exponentiation is also accelerated, via its direct link with Frobenius
form computation~\cite[Thm.\,7.3]{giesbrecht1995nearly}.

\myparagraph{Computational model}
Throughout this paper, $\field$ is an effective field. The cost analyses
consist in giving an asymptotic bound on the number of arithmetic operations in
$\field$ used by the algorithm. The operations are addition, subtraction,
multiplication, and inversion in the field, as well as testing whether a given
field element is zero. We use that two polynomials in $\pring$ of degree
bounded by $d$ can be multiplied in $O(d \log(d) \loglog(d))$
operations in \(\field\)~\cite[Chap.\,8]{GathenGerhard1999}. On rare occasions, mostly in
Section~\ref{sec:polmat}, we use a multiplication time function $d \mapsto \timepm{d}$
for $\pring$
to develop somewhat more general results~\cite[Chap.\,8]{GathenGerhard1999}.
This multiplication time function is subject to some convenient assumptions
(see Section~\ref{sec:polmat:timefunc}), which are satisfied in particular for an
$O(d\log(d) \loglog(d))$ algorithm. 

\myparagraph{Notation}
For an $m\times n$ matrix $A$, we write $a_{ij}$ for its entry $(i,j)$. Given
sets $I$ and $J$ of row and column indices, $A_{I,J}$ stands for the
corresponding submatrix of $A$; we use~$*$ to denote all indices, such as in
$A_{I,*}$ or~$A_{*,J}$.
We manipulate tuples  $d\in (\NN \cup \{-\infty\})^m$ of indices or polynomial
degrees, and write $\sumTuple{d}$ for the sum of the entries. For
a polynomial matrix $A \in \field[x]^{m\times n}$, the column degree $\cdeg{A}$
is the tuple of its column degrees $(\max_{1\leq i \leq m} \deg(a_{i,j}))_{1\leq j \leq n}$.

%%%%%%%%%%%%%%%%%%%%%%%%%%%%%%%%%%%%%%%%%%%%%%%%%%%%%%%%%%%%%%%%%%%%%%%%%%%%%%%%%%%
%
%              SEC: POLYNOMIAL MATRICES
%
%%%%%%%%%%%%%%%%%%%%%%%%%%%%%%%%%%%%%%%%%%%%%%%%%%%%%%%%%%%%%%%%%%%%%%%%%%%%%%%%%%%

\section{Kernel basis, Hermite normal form}
\label{sec:prelimin}

We recall the complexity bounds for two fundamental problems which we rely on:
kernel basis and Hermite normal form. The more technical presentation of some
of the other ingredients needed to manipulate polynomial matrices, such as
``unbalanced'' truncated inversion and multiplication, is deferred to
Section~\ref{sec:polmat}.

\subsection{Minimal kernel basis}
\label{subsec:kernel}

A core tool in \cref{algo:MaxIndices,algo:KrylovMatrix} is the computation of
minimal kernel bases of polynomial matrices \cite[Sec.\,6.5.4,
p.\,455]{Kailath80}. The matrix fraction description in
Eq.~(\ref{eqn:series_expansion_intro}) can indeed be rewritten as   
\[
\begin{bmatrix} I-xA  & -U \end{bmatrix}
\begin{bmatrix} S \\ T \end{bmatrix} = 0.
\]
For a matrix \(F \in \pmatspace{n}{m}\), its (right) kernel is the
\(\pring\)-module formed by the vectors \(p \in \pvecspace{m}\) such that \(Fp
= 0\); it has rank \(m - \rank{F}\). A kernel basis is said to be minimal if it
is column reduced, that is, its leading matrix has full column rank
\cite[Sec.\,6.3, p.\,384]{Kailath80}. An efficient algorithm for minimal kernel
bases was described in~\cite{ZhouLabahnStorjohann12}, and its complexity was
further analyzed in the case of a full rank input \(F\) in
\cite[App.\,B]{JeannerodNeigerSchostVillard2017} and
\cite[Lem.\,2.10]{NeigerPernet21}. We will use the following particular case of
the latter result:

\begin{lemma} \emph{(\cite[Algo.\,1]{ZhouLabahnStorjohann12}, and analyses in 
\cite{JeannerodNeigerSchostVillard2017,NeigerPernet21}.)}
  \label{lem:kernel_basis}
  Let \(F \in \pmatspace{n}{(m+n)}\) have rank \(n\) and degree \(\le 1\), with
  \(m \in \bigO{n}\). There is an algorithm \textproc{MinimalKernelBasis} which,
  on input \(F\), returns a minimal kernel basis \(B\in\pmatspace{(m+n)}{m}\)
  for \(F\) using \(\bigO{n^\expmm}\) operations in~\(\field\). Furthermore,
  \(\sumTuple{\cdeg{B}} \le n\).
\end{lemma}

The complexity bound in \cite{NeigerPernet21} uses a multiplication time
function with assumptions that are satisfied in our case (see
\cite[Sec.\,1.1]{NeigerPernet21}). 
Apart from supporting degree~\(>1\), the three listed references also
consider the more general \emph{shifted} reduced bases \cite{BLV06}. The
non-shifted case is obtained by using the uniform shift \(\shift = (1,\ldots,1)
\in \NN^{m+n}\), for which \(s\)-reduced bases are also (non-shifted) reduced
bases. This shift does satisfy the input requirement \(\shift\ge \cdeg{F}\) and
it has sum \(\sumTuple{\shift} = m+n\);
\cite[Thm.\,3.4]{ZhouLabahnStorjohann12} then guarantees that \(B\) has sum of
\(s\)-column degrees at most \(m+n\), which translates as \(\sumTuple{\cdeg{B}}
\le n\).

\subsection{Hermite normal form}

A nonsingular polynomial matrix \(H \in \pmatspace{m}{m}\) is in (column)
Hermite normal form if it is upper triangular, with monic diagonal entries, and
all entries to the right of the diagonal have degree less than that of the
corresponding diagonal entry: \(\deg(h_{ij}) < \deg(h_{ii})\) for \(1 \le i < j
\le m\). Given a nonsingular matrix \(T \in \pmatspace{m}{m}\), there is a
unique matrix \(H\in \pmatspace{m}{m}\) in Hermite normal form which can be
obtained from~\(T\) via unimodular column operations, meaning \(T V = H\) for
some \(V \in \pmatspace{m}{m}\) with \(\det(V) \in \field\setminus\{0\}\). It
is called the Hermite normal form of \(T\).
In fact, in this paper we only seek the diagonal degrees of \(H\)
(see Section~\ref{sec:find_degrees:H}). Since off-diagonal entries can be reduced by
diagonal ones through unimodular column operations, these diagonal degrees are
the same for all triangular forms unimodularly right equivalent to \(T\).

\begin{lemma}\emph{(\cite[Prop.\,3.3]{LabahnNeigerZhou2017}.)}
  \label{lem:hermite_form}
  Let $T$ be nonsingular in \(\pmatspace{m}{m}\). There is an algorithm
  \textproc{HermiteDiagonal} which takes $T$ as input and returns the diagonal
  entries \((h_{11},\ldots,h_{mm}) \in \pring^m\) of the Hermite normal form of
  \(T\) using \(\bigO{m^{\expmm-1} n \log(n)^\expoH}\) operations in
  \(\field\), where \(n = \max(m,\sumTuple{\cdeg{T}})\) and \(\expoH\) is a
  positive real constant.
\end{lemma}

The algorithm as described in \cite[Algo.\,1]{LabahnNeigerZhou2017} works with
an equivalent lower triangular definition of the form. It involves two main
tools: kernel bases \cite{ZhouLabahnStorjohann12}, whose cost involves a
logarithmic factor of the form \(\log(n/m)^2 \loglog(n/m)\)
\cite[Lem.\,2.10]{NeigerPernet21}, and column bases \cite{ZhouLabahn2013}, for
which there is no analysis of the number of logarithmic factors in the
literature. In Lemma~\ref{lem:hermite_form} we introduce the latter as a constant
\(\expoH>0\), and this remains to be thoroughly analyzed.

%%%%%%%%%%%%%%%%%%%%%%%%%%%%%%%%%%%%%%%%%%%%%%%%%%%%%%%%%%%%%%%%%%%%%%%%%%%%%%%%%%%
%
%              SEC: few vectors
%
%%%%%%%%%%%%%%%%%%%%%%%%%%%%%%%%%%%%%%%%%%%%%%%%%%%%%%%%%%%%%%%%%%%%%%%%%%%%%%%%%%%

\section{Krylov basis via polynomial kernel}
\label{sec:KrylovMatrix}

The key ingredients in our approach for computing Krylov bases transform the
problem into polynomial matrix operations. In this section, we present two
algorithms which involve matrix fraction descriptions as in
Eq.~(\ref{eqn:series_expansion_intro}) or equivalent formulations (see also
Appendix~\ref{sec:alternativeto}). In Section~\ref{sec:find_degrees}, given $A\in
\field^{n\times n}$ and \(U = [u_1 \cdots u_m]\) formed from $m$ vectors in
$\field^n$, Algorithm \texttt{MaxIndices} computes the maximal Krylov indices of
$\orb{A}{U}$. It exploits the fact that these indices coincide with the
degrees of certain minimal polynomial relations between the columns of
$[xI-A \;\; -U]$, allowing the use of polynomial
matrix manipulations. In Section~\ref{sec:compute_krylow_few}, we consider the
problem of computing the Krylov matrix \(\kry{A}{U}{d}\) for a given tuple
$d\in \NN^m$. Based on Eq.~(\ref{eqn:series_expansion_intro}), after computing a
kernel basis to obtain $S$ and~$T$, Algorithm \texttt{KrylovMatrix} proceeds with
a matrix series expansion, using the truncation orders given by $d$.  
Joining both algorithms, from \(A\) and \(U\) one obtains the maximal Krylov
basis of \(\orb{A}{U}\) using $O(n^{\expmm})$ field operations as soon
as the number \(m\) of vectors forming $U$ is in
\(\bigO{n/\log(n)^{\expofinal}}\). This will be done explicitly in Lines
\ref{step:MaxKrylovBasis:expofinal} to \ref{step:MaxKrylovBasis:smallm_return} of
Algorithm \texttt{MaxKrylovBasis} in Section~\ref{sec:KrylovBasis}.

%%%%%%%%%%%%%%%%%%%%%%%%%%%%%%%%%%%%%%%%%%%%%%%%%%%%%%%%%%%%%%%%%%%%%%%%%%%%%%%%%%%
%
%              SUBSEC: HERMITE / FINDING THE DEGREES
%
%%%%%%%%%%%%%%%%%%%%%%%%%%%%%%%%%%%%%%%%%%%%%%%%%%%%%%%%%%%%%%%%%%%%%%%%%%%%%%%%%%%

\subsection{Computing the maximal indices}
\label{sec:find_degrees}

We begin in Section~\ref{sec:find_degrees:KG} with a brief reminder of a useful
characterization of the maximal Krylov indices in terms of linear algebra
over~$\field$. 
We then show how to reduce their computation to operations on polynomial
matrices, by explaining that they coincide with the degrees of certain
polynomials in a kernel basis (Lemma~\ref{lem:hermite_gives_degrees}).

Linear dependencies between vectors in Krylov subspaces are translated into
polynomial relations using the $\field[x]$-module structure of~$\field^n$ based
on $xu=Au$ for $u\in \field ^n$~\cite[Sec.\,3.10]{Jac85}. Some of these
dependencies in \(\orb{A}{U}\) are given by the coefficients of the entries of
triangular matrices $H\in \pmatspace{m}{m}$ such that
\begin{equation}
  \label{eq:matfracrev}
  (xI-A)^{-1}U=LH^{-1}
\end{equation}
with $L\in \pmatspace{n}{m}$~\cite[Sec.\,6.7.1, p.\,476]{Kailath80}. In
particular, if $L$ and $H$ are right coprime then the maximal indices of
\(\orb{A}{U}\) are given by the diagonal degrees of $H$. This is detailed
in Section~\ref{sec:find_degrees:H}, considering equivalently that $\trsp{[\trsp{L} \;\;
\trsp{H}]}$ is a kernel basis of $[xI-A \;\; -U]$.
The computation of the indices follows in Section~\ref{sec:find_degrees:algo}: it
combines the kernel and Hermite form tools discussed in  Section~\ref{sec:prelimin}.

\subsubsection{Keller-Gehrig's branching algorithm}
\label{sec:find_degrees:KG}

The most efficient algorithm so far for computing the lexicographically largest
tuple~$d$ such that \(\kry{A}{U}{d}\) is a Krylov basis is given in
\cite[Sec.\,5]{KeGe85}. For a general~$U$ (with $m\in O(n)$), its cost bound
$O(n^\expmm \log(n))$ is notably due to exponentiating $A$ in order to generate
Krylov iterates.
Keller-Gehrig's approach uses the following characterization of the maximal
indices $d$, which we will need (a proof can be found in
Section~\ref{sec:krylovstuff}): for $1\leq j\leq m$,~$d_j$ is the smallest integer such
that
\begin{equation}
  \label{eq:characd}
  A^{d_j}u_j \in \operatorname{Span}(u_j, Au_j, \ldots, A^{d_j-1} u_j) +\orb{A}{U_{*,1..j-1}}.
\end{equation}
A recursive construction allows independent and increasingly long Krylov chains
$u_j, Au_j, \ldots, A^{\ell-1} u_j, \ell \geq 0$ to be joined together in order
to reach the maximal basis of $\orb{A}{U}$, with independence guaranteed by
Gaussian elimination \cite[Sec.\,4]{KeGe85}. We will return to this in our
general Krylov basis method in Section~\ref{sec:KrylovBasis}.

\subsubsection{Links with the Hermite normal form}
\label{sec:find_degrees:H}

The following is to be compared with known techniques for  matrix fraction
descriptions~\cite[Sec.\,6.4.6, p.\,424]{Kailath80}, or to a formalism
occasionally used to efficiently compute matrix normal forms~[\citealp{Vil97};
\citealp[Chap.\,9]{Sto00}].

\begin{lemma}
  \label{lem:hermite_gives_degrees} 
  Given $A \in \matspace{n}{n}$ and $U \in \matspace{n}{m}$, let
  $\trsp{[\trsp{L} \;\; \trsp{H}]}$ be a kernel basis
  of $[xI-A \;\; -U]$ such that $H\in
  \pmatspace{m}{m}$ is upper triangular. The matrix $H$ is nonsingular and its
  diagonal degrees  are the maximal Krylov indices of \(\orb{A}{U}\).  
\end{lemma}
\begin{proof}
  We first note that, for an upper triangular $P = [p_{i,j}]_{i,j}$ in
  $\pmatspace{m}{m}$, $(xI-A)^{-1}UP$ is a polynomial matrix if and only if
  \begin{equation}
    \label{eq:ddsum}
    p_{jj}(A) u_j + \sum_{i=1}^{j-1} p_{ij}(A) u_i=0, 1\leq j \leq m.
  \end{equation}
  Indeed, using expansions in \(1/x\) with $(xI-A)^{-1}= \sum _{k\geq 0}
  A^k/x^{k+1}$, the coefficient of $1/x^{k+1}$ in the $j$th column of
  $(xI-A)^{-1}UP$ is
  \begin{equation}
    \label{eqn:coefficients_ddsum}
    \sum_{\ell=0}^{\deg(p_{jj})} A^{\ell+k} u_j p_{jj}^{(\ell)} + \sum_{i=1}^{j-1} \sum_{\ell=0}^{\deg(p_{ij})} A^{\ell+k} u_i p_{ij}^{(\ell)},
  \end{equation}
  for any \(k \ge 0\). Here $p_{ij}^{(\ell)}$ is the coefficient of degree
  $\ell$ of $p_{i,j}$, for $1\leq i \leq j \leq m$. The ``if'' direction
  follows: taking $k=0$, the (expansion of the) polynomial matrix
  $(xI-A)^{-1}UP$ has no term in $1/x$, hence \cref{eq:ddsum}. Conversely, if
  \cref{eq:ddsum} holds, multiplying it by $A^k$ shows that the expression in
  Eq.~(\ref{eqn:coefficients_ddsum}) is zero, for any $k\geq 0$; hence the terms in
  $1/x^{k+1}$ are zero in the expansion of $(xI-A)^{-1}UP$, which is thus a
  polynomial matrix.

  Now let $d$ be the maximal indices, and let $d'$ be the tuple of the diagonal
  degrees of~$H$. Observe that $H$ is nonsingular: since $L=(xI-A)^{-1}UH$,
  $Hu=0$ for $u\neq 0$ would lead to $\trsp{[\trsp{L} \;\; \trsp{H}]}u=0$,
  which is impossible by definition of bases. Thus \(d'\in\NN^m\).

  To conclude, we first show that $d$ is greater than or equal to $d'$, and
  then the converse. The characterization of $d$ given by \cref{eq:characd}
  leads to linear dependencies as in \cref{eq:ddsum}, hence to~$P$ with
  diagonal degrees given by $d$. Since $(xI-A)^{-1}UP$ is a polynomial matrix,
  say~$R$, the columns of $\trsp{[\trsp{R} \;\; \trsp{P}]}$ are in the kernel
  of $[xI-A \;\; -U]$, which is generated by $\trsp{[\trsp{L} \;\; \trsp{H}]}$,
  hence we must have $P=HQ$ for some polynomial matrix $Q$.  Since both~$P$ and
  $H$ are triangular we conclude that $d \ge d'$ entry-wise. On the other hand,
  considering $H$, we get dependencies as in \cref{eq:ddsum} with polynomials
  given by the entries of $H$. Thanks to the fact that the $d_j$'s reflect the
  shortest dependencies (see Section~\ref{sec:find_degrees:KG}), we get that $d' \ge
  d$ entry-wise.
\end{proof}

\subsubsection{Algorithm \texttt{MaxIndices}}
\label{sec:find_degrees:algo}

Let $V\in \pmatspace{m}{m}$ be unimodular such that $H=TV$ is in Hermite form.
The matrix $\trsp{[\trsp{(SV)} \;\; \trsp{H}]} = \trsp{[\trsp{S} \;\;
\trsp{T}]}V$ is also a kernel basis of $[xI-A \;\; -U]$, and the correctness of
the algorithm follows from Lemma~\ref{lem:hermite_gives_degrees} with \(L = SV\).

For the complexity bound, the kernel basis at Line~\ref{step:chains:kernel} costs
$O(n^{\expmm})$ operations in \(\field\) when $m\in O(n)$, according to
Lemma~\ref{lem:kernel_basis}, which also ensures \(\sumTuple{\cdeg{\mat{T}}} \le
n\).  From Lemma~\ref{lem:hermite_gives_degrees}, the resulting matrix $T$ is
nonsingular because its Hermite form is. Lemma~\ref{lem:hermite_form} states that
Line~\ref{step:chains:hermite} can be achieved using $O(m^{\expmm-1} n
\log(n)^\expoH)$ operations in $\field$. In particular, as soon as the number
\(m\) of columns of $U$ is in \(\bigO{n/\log(n)^{{\expoH}/({\expmm-1})}}\), the
overall complexity bound is $O(n^{\expmm})$.

\begin{algorithm}[ht]
  \algoCaptionLabel{MaxIndices}{A,U}
  \begin{algorithmic}[1]
    \Require \(\mat{A} \in \matspace{n}{n}, \mat{U} \in \matspace{n}{m}\)
    \Ensure the tuple in $\NN^m$ of the maximal indices of \(\orb{A}{U}\) 

    \State\label{step:chains:kernel}%
    \CommentLine{minimal kernel basis \cite[Algo.\,1]{ZhouLabahnStorjohann12}}
    \Statex \([\begin{smallmatrix} S \\ T \end{smallmatrix}] \gets \textproc{MinimalKernelBasis}([xI-A \;\; -U])\)
    \Statex  where \(S \in \pmatspace{n}{m}\) and \(T \in \pmatspace{m}{m}\) 

    \State\label{step:chains:hermite}%
    \CommentLine{diagonal entries of the Hermite normal form}
    \Statex \((h_{11},\ldots, h_{mm}) \in \pring^m \gets \textproc{HermiteDiagonal}(T)\) \Comment{\cite[Alg.\,1]{LabahnNeigerZhou2017}}

    \State \Return \((\deg(h_{11}),\ldots,\deg(h_{mm}))\)
  \end{algorithmic}
\end{algorithm}

\subsection{Computing a Krylov matrix}
\label{sec:compute_krylow_few}

Hereafter, for a column polynomial vector \(v = \sum_j v_j x^j \in
\pvecspace{m}\) and an integer \(d\in\NN\), we denote its truncation at order
\(d\) by
\[
  v \rem x^d  = v_0 + v_1 x + \cdots + v_{d-1} x^{d-1} \in
  \pvecspace{m}
\]
and the corresponding matrix of coefficients by
\[
  \coeffs{v}{d}
  = [v_0 \;\; v_1 \;\; \cdots \;\; v_{d-1}] \in \matspace{m}{d}.
\]
We extend this column-wise for \(M \in \pmatspace{m}{k}\) and \(d = (d_j)_j\in
\NN^k\), that is, \(M \rem x^d = [\matcol{M}{1} \rem x^{d_1} \;\;\cdots\;\;
\matcol{M}{k} \rem x^{d_k}]\).

Algorithm \texttt{KrylovMatrix} uses polynomial matrix subroutines at
Lines~\ref{step:kryloveval:inverse} and \ref{step:kryloveval:matmul}. They are handled in
detail in Sections~\ref{sec:polmat:TruncatedInverse} and \ref{sec:polmat:TruncatedProduct}.

\begin{lemma}
  \label{lem:KrylovMatrix}
  Algorithm \texttt{KrylovMatrix} is correct. If \(m\) and \(\sumTuple{d}\) are
  both in \(\bigO{n}\), it uses \(\bigO{n^\expmm + m^{\expmm-2} n^2
  \log(n)^4}\) operations in \(\field\).
\end{lemma}
\begin{proof}
  Knowing by Lemma~\ref{lem:kryloveval:constant_invertible} that \(T(0)\) is
  invertible (hence \(T\) is nonsingular), the correctness follows from the
  series expansion
  \[
    ST^{-1} = (I - x A)^{-1} U = \sum_{k\ge 0} x^k A^k U
  \]
  mentioned in Eq.~(\ref{eqn:series_expansion_intro}); here the first equality comes
  from the kernel relation $[I-xA \;\; -U] \trsp{[\trsp{S} \;\; \trsp{T}]}=0$,
  by construction in Line~\ref{step:kryloveval:kernel}.  Since \(T(0)\) is
  invertible,
  Lines~\ref{step:kryloveval:inverse} and \ref{step:kryloveval:matmul} compute \(P = SQ \rem
  x^d = ST^{-1} \rem x^d\), according to
  \cref{prop:TruncatedInverse,prop:TruncatedProduct}. The above equation yields
  \(\coeffs{\matcol{P}{j}}{d_j} = [u_j \;\; Au_j \;\; \cdots \;\;
  A^{d_j-1}u_j]\), hence the correctness of the output formed at
  Line~\ref{step:kryloveval:expand}. 

  We turn to the complexity analysis. We assume \(m \le \sumTuple{d}\) without
  loss of generality: the columns of \(U\) corresponding to indices \(j\) with
  \(d_j=0\) could simply be ignored in Algorithm \texttt{KrylovMatrix}, reducing
  to a case where all entries of \(d\) are strictly positive.

  By Lemma~\ref{lem:kernel_basis}, Line~\ref{step:kryloveval:kernel} uses
  \(\bigO{n^\expmm}\) operations in \(\field\) and ensures
  \(\sumTuple{\cdeg{K}} \le n\). In particular, \(\sumTuple{\cdeg{S}} \le n\)
  and \(\sumTuple{\cdeg{T}} \le n\).

  The latter bound ensures that the generic determinantal degree \(\Delta(T)\)
  is in \(\bigO{n}\) (see Lemma~\ref{lem:partial_linearization}), hence we can use
  the second cost bound in \cref{prop:TruncatedInverse}:
Line~\ref{step:kryloveval:inverse} computes \(Q = T^{-1} \rem x^d\) using 
  \begin{equation}
    \label{eqn:comp_trunc_inv}
    \bigO{m^\expmm \timepm{\frac{n}{m}}
      \left( \log(n) + \left\lceil \frac{\sumTuple{d}}{n} \right\rceil \log(m) \log(\sumTuple{d}) \right)
    }
  \end{equation}
  field operations. The second cost bound in \cref{prop:TruncatedProduct}
  states that Line~\ref{step:kryloveval:matmul} computes \(SQ \rem x^d\) using
  \(\bigO{m^{\expmm-2} n \timepm{n+\sumTuple{d}}}\) operations in \(\field\).
  Summing the latter bound with that in Eq.~(\ref{eqn:comp_trunc_inv}) yields a cost
  bound for Algorithm \texttt{KrylovMatrix} when the only assumption on \(d\) is
  \(m \le \sumTuple{d}\). Now, assuming further \(\sumTuple{d} \in \bigO{n}\),
  the latter bound becomes \(\bigO{m^{\expmm-2} n \timepm{n}}\), whereas the
  one in Eq.~(\ref{eqn:comp_trunc_inv}) simplifies as $\bigO{m^\expmm
  \timepm{\frac{n}{m}} \log(n) \log(m)}$. The claimed cost bound then follows
  from \(\timepm{n} \in \bigO{n \log(n) \loglog(n)}\).
\end{proof}

\renewcommand{\thealgorithm}{2}
\begin{algorithm}[ht]
  \algoCaptionLabel{KrylovMatrix}{A,U,d}
  \begin{algorithmic}[1]
    \Require \(\mat{A} \in \matspace{n}{n}\), \(\mat{U} \in \matspace{n}{m}\), \(d = (d_1,\ldots,d_m) \in \NN^m\)
    \Ensure the Krylov matrix \(\kry{A}{U}{d} \in \matspace{n}{(d_1 + \cdots + d_m)}\)

  \State \label{step:kryloveval:kernel}
    \CommentLine{minimal kernel basis \cite[Algo.\,1]{ZhouLabahnStorjohann12}}
  \Statex \([\begin{smallmatrix} S \\ T \end{smallmatrix}] \gets \textproc{MinimalKernelBasis}([I-xA \;\; -U])\)
    \Statex where \(S \in \pmatspace{n}{m}\) and \(T \in \pmatspace{m}{m}\) 

  \State \label{step:kryloveval:inverse}
    \CommentLine{column-truncated inverse \(T^{-1} \rem x^d\), detailed in Section~\ref{sec:polmat:TruncatedInverse}}
    \Statex \(Q \in \pmatspace{m}{m} \gets \Call{algo:TruncatedInverse}{T,d}\)

  \State \label{step:kryloveval:matmul}
    \CommentLine{column-truncated product \(S Q \rem x^d\), detailed in Section~\ref{sec:polmat:TruncatedProduct}}
  \Statex \(P \in \pmatspace{n}{m} \gets \Call{algo:TruncatedProduct}{S,Q,d}\)

  \State \label{step:kryloveval:expand}
    \CommentLine{linearize columns of \(P\) into a constant matrix and return}
  \Statex \Return \([\coeffs{\matcol{P}{1}}{d_1} \; | \; \cdots \; | \; \coeffs{\matcol{P}{m}}{d_m}] \in \matspace{n}{\sumTuple{d}}\)
  \end{algorithmic}
\end{algorithm}

\begin{lemma}
  \label{lem:kryloveval:constant_invertible}
  For \(T\) as in Line~\ref{step:kryloveval:kernel} of \cref{algo:KrylovMatrix},
  \(T(0)\) is invertible.
\end{lemma}

\begin{proof}
  Let $B=\trsp{[\trsp{S} \;\; \trsp{T}]}$. As a kernel basis, \(B\) can be
  completed into a basis of \(\pring^{m+n}\) \cite[Lem.\,7.4, p148]{Lang2002}:
  there is \(C \in \pmatspace{(m+n)}{n}\) such that \(\det([B \;\; C]) =
  \det([B(0) \;\; C(0)]) = 1\). We conclude by deducing that any vector \(v \in
  \vecspace{m}\) such that \(\mat{T}(0) v = 0\) must be zero: from \([I \;\;
  -U] B(0) = 0\) we get \(S(0) v = 0\), hence \(B(0) v = 0\), which implies
  \(v=0\) since \([B(0) \;\; C(0)]\) is invertible.
\end{proof}
\enlargethispage{12pt}

\section{Preprocessing of small indices}
\label{sec:KrylovBasis}

In this section we consider the case where the number \(m\) of vectors to be
iterated can be large. Say for instance \(m=\Theta(n)\): then, terms with
logarithmic factors in the costs of~\cref{algo:KrylovMatrix,algo:MaxIndices},
such as \(O({m^{\expmm-1}n \log(n)^\expoH})\) or
\(O({m^{\expmm-2}n^2\log(n)^4})\), are not negligible anymore compared to the
cost \(\bigO{n^\expmm}\) of the kernel basis computations.

Algorithm \texttt{MaxKrylovBasis} computes the maximal Krylov basis of
$\orb{A}{U}$. It involves a preprocessing phase based on the Keller-Gehrig
branching algorithm~\cite[Thm~5.1]{KeGe85}, which is terminated after a small
number \(\DD\) of iterations. This ensures that after this phase, no more than
\(m=n/2^\DD\) vectors need to be iterated further, which can then be performed
by \cref{algo:MaxIndices,algo:KrylovMatrix}. Specifically, setting
\(\ell=\lceil c\loglog(n)\rceil\) with \(\expofinal = \max(4/(\expmm-2),
c_1/(\expmm-1))\) yields \(m \approx n/\log(n)^c\), ensuring that
\cref{algo:MaxIndices,algo:KrylovMatrix} run in~\(\bigO{n^\expmm}\).

If $m$ is sufficiently small from the start, the preprocessing phase is
skipped: a direct call to the latter algorithms is made. Otherwise, a
\textbf{for } loop (Line~\ref{step:preloop} of Algorithm \texttt{MaxKrylovBasis})
performs $\DD$ loop iterations akin to those used in \cite[Sec.\,5]{KeGe85}. At
the end of this loop, the indices \(J\) of the vectors requiring further
iterations are identified (Line~\ref{step:longer}), and processed via
\cref{algo:MaxIndices,algo:KrylovMatrix} (Line~\ref{step:HCL2} and \ref{step:kmat}). Finally,
both temporary sequences of iterates, the ``long'' ones (indices in $J$) and
``short'' ones (indices not in $J$), are merged via Gaussian elimination and
the maximal basis is obtained.

\renewcommand{\thealgorithm}{3}
\begin{algorithm}[htb]
  \algoCaptionLabel{MaxKrylovBasis}{A,U}
  \begin{algorithmic}[1]
    \Require \(A\in \field^{n\times n}\), 
    \(U=[u_1 \;\; \cdots \;\; u_m] \in \field^{n\times m}\)
    \Ensure the maximal Krylov basis of $\orb{A}{U}$

    \State \(\expofinal\gets \max(4/(\expmm-2),c_1/(\expmm-1))\);
            \(\thres \gets \log_2(n)^c\)
        \label{step:MaxKrylovBasis:expofinal}
    \If{\(m\leq n/\thres\)}
      \State \(d \gets \Call{algo:MaxIndices}{A,U}\)
          \label{step:HCL1} 
      \State \(K\gets \Call{algo:KrylovMatrix}{A,U,d}\); \Return \(K\)
          \label{step:MaxKrylovBasis:smallm_return}
    \EndIf

    \State \CommentLine{Preprocessing phase in a Keller-Gehrig fashion \cite{KeGe85}}
    \State \(\DD \gets \lceil \log_2(t)\rceil\)
    \State \(V^{(0)} \gets U; \ \delta \gets (1,\ldots, 1)\in \ZZ^m\)
    \State \(B \gets A\)
    \For{\(i\gets 0, \ldots, \DD-1\)} \label{step:preloop}
      \State Write \(V^{(i)} = [V^{(i)}_1 \;\; V^{(i)}_2 \;\;\cdots \;\; V^{(i)}_m]\) where 
          \label{step:mkb:invariant} 
          \StateX{2} \(V^{(i)}_j = \kry{A}{U_{*,j}}{\delta_j} \in \field^{n\times \delta_j}\)
      \State \(J =\{j_1<\cdots<j_s\} \gets \{j\in \range{m} \mid \delta_{j}=2^i\}\)

      \State \([W^{(i)}_{j_1} \;\; \cdots \;\; W^{(i)}_{j_s}] \gets
                B [V^{(i)}_{j_1} \;\; \cdots \;\; V^{(i)}_{j_s}]\)   

      \State \(W_j^{(i)}\gets []\in \field^{n\times 0} \) for \(j\notin J\)
      \State \(Z \gets [V^{(i)}_1 \;\; W^{(i)}_1 \;\; \cdots \;\; V^{(i)}_m \;\; W^{(i)}_m]\)
      \State \(\mathcal{C}\gets \text{ColRankProfile}(Z)\)
            \label{step:CRP1} 
      \State \(\delta\gets (\delta_1,\ldots,\delta_m)\) such that \(\delta_j\) is maximal with 
            \label{step:maxdelta} 
            \StateX{2} \(\kry{A}{u_j}{\delta_j} = Z_{*,\mathcal{C} \cap \{b..b+\delta_j-1\}}\) for some \(b\)
      \State  \(V^{(i+1)} \gets [V_1^{(i+1)} \;\; \cdots  \;\; V_m^{(i+1)}]\) where
            \StateX{2} each \(V_j^{(i+1)}\gets\kry{A}{u_j}{\delta_j}\) is copied from \(Z\)
      \State \(B\gets B^2\)
    \EndFor
    \State \CommentLine{Here, \(V^{(\DD)}=\kry{A}{U}{\delta}\) and \(\delta=(\delta_1,\ldots, \delta_m) \in \{0,\ldots,2^\DD\}^m\) is 
        \label{step:loopexit} 
        \StateX{2} lexicographically maximal such that \(\kry{A}{U}{\delta}\) has full rank
    }
    \State \CommentLine{Further iterations for selected vectors, using polynomial matrices}
    \State \(J =\{j_1<\cdots<j_s\} \gets \{j\in \range{m}, \delta_j=2^\DD\}\)
          \label{step:longer}
    \State \(d\gets \Call{algo:MaxIndices}{A,\matcols{U}{J}}\)
          \label{step:HCL2}
    \State \([K_{j_1} \;\; \cdots \;\; K_{j_s}] \gets \Call{algo:KrylovMatrix}{A,\matcols{U}{J},d}\)
          \label{step:kmat}
    \State \CommentLine{Final merge}
    \State \(K_j \gets V_j^{(\DD)}\) for \(j\notin J\);
            \(\;\;K\gets [K_1 \;\; \cdots \;\; K_m]\) \label{step:beforemerge}
    \State \Return \(\matcols{K}{\text{ColRankProfile}(K)}\)
          \label{step:return} 
\end{algorithmic}
\end{algorithm}

\begin{theorem}
  \label{thm:maxkrylovbasis}
  Algorithm \texttt{MaxKrylovBasis} is correct. If \(m=\bigO{n}\), it uses
  \(\bigO{n^\expmm\loglog(n)}\) operations in \(\field\), and if
  \(m=\bigO{n/\log(n)^{\expofinal}}\) for the constant \(\expofinal > 0\) in
  Line~\ref{step:MaxKrylovBasis:expofinal}, it uses \(\bigO{n^\expmm}\) operations
  in \(\field\).
\end{theorem}
\begin{proof}
As in Keller-Gehrig's branching algorithm~\cite{KeGe85}, the loop invariant is
\(V^{(i)}=\kry{A}{U}{\delta}\) with \(\delta\in\{0,\ldots,2^i\}^m\)
lexicographically maximal such that \(V^{(i)}\) has full rank.
By induction, it remains valid upon exiting the loop, as stated
in~Line~\ref{step:loopexit}. At this point, \(s =\card{J}\) is the number of
vectors left to iterate. Indeed, \(j\notin J\) means that a linear relation of
the type
\begin{equation}
  \label{eq:characdproof}
  A^{\delta _j}u_j \in \operatorname{Span} (u_j, Au_j, \ldots, A^{\delta _j-1} u_j) +\orb{A}{U_{*,1..j-1}}
\end{equation}
has already been found (recall \cref{eq:characd}).

A maximal Krylov basis for the subset of vectors given by $J$ is then computed
using the algorithms of~Section~\ref{sec:KrylovMatrix}.
All the Krylov iterates that form the matrix \(K\) at Line~\ref{step:beforemerge} are
considered at orders higher than or equal to the final maximal ones, because relations of
the type \cref{eq:characdproof} have been detected for all $j$. Thus, a column
rank profile computation yields the maximal indices and selects the
vectors for the maximal basis.

The choice for the parameter \(\thres\) ensures that the for-loop is executed
using \(\bigO{n^\expmm\loglog(n)}\) field operations. On the other hand,
\cref{algo:KrylovMatrix,algo:MaxIndices} are called with \(m=s\leq n/2^\DD \leq
n/\log(n)^{\expofinal}\), and therefore run in \(\bigO{n^\expmm}\) field
operations.
\end{proof}

In addition, Algorithm \texttt{MaxKrylovBasis} can be adapted so as to compute a
Krylov matrix \(\kry{A}{U}{d}\) for given orders \(d\) (not necessarily the
maximal indices) that are provided as additional input.

\begin{corollary}
Given \(A\in \field^{n\times n}\), \(U\in\field^{n\times m}\), and
\(d\in\NN^m\), with $m$ and $|d|$ both in \(\bigO{n}\), the Krylov matrix
\(\kry{A}{U}{d}\) can be computed using \(\bigO{n^\expmm\loglog(n)}\)
operations in \(\field\), or \(\bigO{n^\expmm}\) operations in \(\field\) if
\(m \in \bigO{n/\log(n)^{{\expofinal}}}\) where $\expofinal$ is a positive
real constant. 
\end{corollary}

This only requires the following modifications:
\begin{enumerate}
  \item Line~\ref{step:CRP1} and \ref{step:maxdelta} should be replaced by
    \begin{algorithmic}
      \State\InlineFor{all \(j\in J\)}{ \(\delta_j \gets \min(2\delta_j,d_j)\)}
      %%\ForAll{\(j\in J\)}
      %%  \State \(\delta_j \gets \min(2\delta_j,d_j)\)
      %%\EndFor
    \end{algorithmic}
  \item Line~\ref{step:return} should be replaced by
    \begin{algorithmic}
      \State \Return \(K\)
    \end{algorithmic}
  \item remove Lines~\ref{step:HCL1} and \ref{step:HCL2}.
\end{enumerate}

%%%%%%%%%%%%%%%%%%%%%%%%%%%%%%%%%%%%%%%%%%%%%%%%%%%%%%%%
%
%  SEC: Frobenius
%
%%%%%%%%%%%%%%%%%%%%%%%%%%%%%%%%%%%%%%%%%%%%%%%%%%%%%%%%
\section{Frobenius and Kalman forms} \label{sec:decomp}

In this section we discuss some consequences of our new algorithms for Krylov
bases, namely some improved complexity bounds for directly related problems.

\subsection{Frobenius normal form}

Generically, the Frobenius normal form of an $n\times n$ matrix can be computed
using $O(n^\expmm)$ operations in $\field$ [\citealp[Sec.\,6]{KeGe85};
\citealp{PernetStorjohann2007}], and the approach in
\cite{PernetStorjohann2007} mainly provides a Las Vegas probabilistic algorithm
in \(\bigO{n^\expmm}\). It is still an open question to obtain the same
complexity bound with a deterministic algorithm, and also to compute an
associated transformation matrix. Our results allow us to make some progress on
both aspects.

Since this is often a basic operation for these problems, we can already note
that from Lemma~\ref{lem:hermite_gives_degrees} and its proof,  the minimal
polynomial of a vector $u\in \field ^n$ can be computed in \(\bigO{n^\expmm}\),
via Lemma~\ref{lem:kernel_basis}. Indeed this minimal polynomial is the last
entry, made monic, of any kernel basis (in this case a single vector) of $[xI-A
\;\; -u]$.

Our algorithms make it possible to obtain \(\bigO{n^\expmm}\) for the Frobenius
form with associated transformation matrix in a special case.  The general cost
bound \(\bigO{n^\expmm\loglog(n)}\) is achieved in~\cite[Theorem
7.1]{PernetStorjohann2007Frob} with a probabilistic algorithm.  To have  a
transformation, we  can first compute the Frobenius form alone using
$O(n^{\expmm})$ operations.  If it has $m$ non-trivial blocks and $U\in \field
^{n\times m}$ is chosen uniformly at random, then we know that a transformation
matrix can be computed from the maximal Krylov basis of \(\orb{A}{U}\) using
$O(n^\expmm)$ operations~\cite[Thm.\,2.5 \& 4.3]{giesbrecht1995nearly}.  So we
have the following. 
\begin{corollary}
  \label{cor:frobproba} 
  Let \(A\in\field^{n\times n}\) with \(\# \field \geq n^2\), and assume that
  its Frobenius normal form has \(m\in\bigO{n/\log(n)^\expofinal}\) non-trivial
  blocks. A transformation matrix to the form can be computed by a Las Vegas
  probabilistic algorithm using \(\bigO{n^\expmm}\) operations in \(\field\). 
\end{corollary}

The fastest deterministic algorithm to compute a transformation to Frobenius
form is given in \cite{Sto01} (see also \cite[Chap.\,9]{Sto00}), with a cost of
\(\bigO{n^\expmm\log(n)\loglog(n)}\). This cost is essentially $O(\loglog(n))$
computations of maximal Krylov bases, plus $O(n^\expmm)$ operations.

\begin{corollary}
  \label{cor:frobdet}
  Given $A\in \field^{n\times n}$, a transformation matrix to the Frobenius
  normal form of \(A\) can be computed using \(\bigO{n^\expmm\loglog(n)^2}\)
  operations in \(\field\).
\end{corollary}

Finally, once the Frobenius form and a transformation matrix are known, and given an integer $k\geq 0$,
computing $A^k$ costs $O(n^\expmm)$ plus $O(\log(k)\timepm{n})$ field
operations (see e.g.\ \cite[Cor.\,7.4]{giesbrecht1995nearly}). As mentioned in
Section~\ref{sec:intro}, here $d \mapsto \timepm{d}$ is a multiplication time function
for $\field[x]$. So, in total, computing \(A^k\) costs
\(\bigO{n^\expmm\loglog(n)^2}\) if $\log(k)\in
\bigO{n^{\expmm-1-\varepsilon}}$. The evaluation of a polynomial $p\in
\field[x]$ at $A$ can also be considered in a similar way, using, for example,
the analysis of \cite[Thm.\,7.3]{giesbrecht1995nearly}.

\subsection{Kalman decomposition}

The study of the structure of linear dynamical systems in control theory is
directly related to Krylov spaces and matrix polynomial forms~\cite{Kailath80}.
For example, the link we use between maximal indices and Hermite form degrees
originates from this correspondence.

Our work could be continued to show that the complexity bound
\(\bigO{n^\expmm\loglog(n)}\) in \cref{thm:maxkrylovbasis} could be applied to
the computation of a Kalman decomposition~[\citealp{kal63};
\citealp[Sec.\,2.4.2, p.\,128]{Kailath80}]. This is beyond the scope of this
paper, so we will not go into detail about it here.  
However, we can specify the main ingredient. Given $A$ and~$U$ with
$\dim(\orb{A}{U})=\nu$, we want to transform the system $(A,U)$ according to
\cite[Sec.\,2.4.2, Eq.\,(11), p.\,130]{Kailath80}: 
\[
  P^{-1}AP=  \begin{bmatrix} A_c & A_1 \\ \matz & A_2 \end{bmatrix}, ~~
  P^{-1}U = \begin{bmatrix} U_c \\ { 0} \end{bmatrix},
\]
where $A_c$ is $\nu\times \nu$, $U_c$ is $\nu \times m$, and $P$ is invertible
in $\field^{n\times n}$. The matrix $P$ can be formed by a Krylov basis of
$\orb{A}{U}$ and a matrix with $n-\nu$ independent columns not in $\orb{A}{U}$.
The general decomposition is obtained by combining  a constant number of such
transformations and basic matrix operations to decompose $(A,U)$.

\section{Polynomial Matrix subroutines}
\label{sec:polmat}

We now detail the subroutines used in Algorithm \texttt{KrylovMatrix}.
Section~\ref{sec:polmat:timefunc} introduces some convenient notation for complexity
bounds. Section~\ref{sec:polmat:TruncatedInverse} combines high-order lifting
\cite{Storjohann2003} and partial linearization~\cite{GuSaStVa12} to compute
truncated inverse expansions, while Section~\ref{sec:polmat:TruncatedProduct} focuses
on truncated matrix products. In both cases, the difficulty towards efficiency
lies in the presence of unbalanced degrees and unbalanced truncation orders.

%%%%%%%%%%%%%%%%%%%%%%%%%%%%%%%%%%%%%%%%%%%%%%%%%%%%%%%%%%%%%%%%%%%%%%%%%%%%%%%%%%%
%
%              POLMAT TOOLS --> COMPLEXITY NOTATION MM(.) etc.
%
%%%%%%%%%%%%%%%%%%%%%%%%%%%%%%%%%%%%%%%%%%%%%%%%%%%%%%%%%%%%%%%%%%%%%%%%%%%%%%%%%%%
\subsection{Complexity helper functions}
\label{sec:polmat:timefunc}

We briefly recall notation and assumptions about cost functions; for more
details, we refer to \cite[Chap.\,8]{GathenGerhard1999} for the general
framework, and to \cite[Sec.\,2]{Storjohann2003} and
\cite[Sec.\,1.1]{NeigerPernet21} for polynomial matrices specifically.

In what follows, we assume fixed multiplication algorithms:
\begin{itemize}
  \item for polynomials in \(\pring\), with cost function \(\timepm{d}\) when
    the input polynomials have degree at most \(d\);
  \item for matrices in \(\matspace{m}{m}\), with cost \(\bigO{m^\expmm}\);
  \item for polynomial matrices in \(\pmatspace{m}{m}\), with cost function
    \(\timepmm{m}{d}\) when the input matrices have degree at most \(d\).
\end{itemize}

To simplify our analyses and the resulting bounds, we make the same assumptions
as in the above references. In particular, \(\expmm>2\), \(\timepm{\cdot}\) is
superlinear, and \(\timepmm{m}{d} \in \bigO{m^\expmm \timepm{d}}\).

We will also use the function \(\timepmmrec{m}{d}\) from
\cite[Sec.\,2]{Storjohann2003}; as noted in this reference, the above-mentioned
assumptions imply \(\timepmmrec{m}{d} \in \bigO{m^\expmm \timepm{d} \log(d)}\).

%%%%%%%%%%%%%%%%%%%%%%%%%%%%%%%%%%%%%%%%%%%%%%%%%%%%%%%%%%%%%%%%%%%%%%%%%%%%%%%%%%%
%
%              POLMAT TOOLS --> TRUNCATED INVERSION
%
%%%%%%%%%%%%%%%%%%%%%%%%%%%%%%%%%%%%%%%%%%%%%%%%%%%%%%%%%%%%%%%%%%%%%%%%%%%%%%%%%%%
\subsection{Polynomial matrix truncated inverse}
\label{sec:polmat:TruncatedInverse}

In Algorithm \texttt{KrylovMatrix}, Line~\ref{step:kryloveval:inverse} asks to find
terms of the power series expansion of \(P^{-1}\), for an \(m\times m\)
polynomial matrix \(P\) with \(P(0)\) invertible. Customary algorithms for this
task, depending on the range of parameters (\(m\), \(\deg(P)\), truncation
order), include a matrix extension of Newton iteration [\citealp{MoCa79}; \citealp[Chap.\,9]{GathenGerhard1999}],
or matrix inversion \cite{ZLS15} followed by Newton iteration on the individual entries.

Here, a first obstacle towards efficiency comes from the heterogeneity of
truncation orders: one seeks the first \(d_j\) terms of the \(j\)th column of
the expansion of \(P^{-1}\), for some prescribed \(d = (d_j)_j \in \NN^m\)
which may have unbalanced entries. In the extreme case \(d =
(d_1,0,\ldots,0)\), the task becomes the computation of many initial
terms of the expansion of \(P^{-1} \trsp{[1 \; 0 \; \cdots \; 0]}\), the
first column of~\(P^{-1}\). This is handled efficiently via high-order lifting
techniques \cite[Sec.\,9]{Storjohann2003}. Our solution for a general tuple
\(d\) is to rely on cases where the high-order lifting approach is efficient,
by splitting the truncation orders into subsets of the type \(\{j \in \range{m}
\mid 2^{k-1} \sumTuple{d}/m < d_j \le 2^{k} \sumTuple{d}/m\}\), for only
logarithmically many values of \(k\). Observe that this subset has cardinality
less than \(m/2^{k-1}\): higher truncation orders involve fewer columns of the
inverse.

A second obstacle is due to the heterogeneity of the degrees in the matrix
\(P\) itself. In the context of Algorithm \texttt{KrylovMatrix}, \(P\) may have
unbalanced column degrees, but they are controlled to some extent: their sum is
at most \(n\) (the dimension of the matrix \(A\)). Whereas such cases were not
handled in the original description of high-order lifting, the partial
linearization tools described in \cite[Sec.\,6]{GuSaStVa12} allow one to deal
with this obstacle. For example, this was applied in
\cite[Lem.\,3.3]{NeigerVu2017}, yet in a way that is not efficient enough
for the matrices \(P\) encountered here: this reference targets low average row degree for \(P\),
whereas here our main control is on the average column degree.
Here, following this combination of \cite[Sec.\,9]{Storjohann2003}
and~\cite[Sec.\,6]{GuSaStVa12}, we present an algorithm which supports a more
general unbalancedness of degrees of \(P\).
In the next two lemmas, we summarize the properties we will use from the
latter references.

\begin{lemma}[{\cite[Sec.\,6]{GuSaStVa12}}]
  \label{lem:partial_linearization}
  Let \(\mat{P} \in \pmatspace{m}{m}\) be nonsingular. Consider its so-called
  generic determinantal degree \(\Delta(P)\),
  \[
    \Delta(P) = \max_{\pi \in \mathfrak{S}_m} \sum_{\substack{1\le i \le m \\ A_{i,\pi_i} \neq 0}} \deg(A_{i,\pi_i})
    \le \sumTuple{\cdeg{P}}.
  \]
  One can build, without using field operations, a matrix \(\mat{\bar{P}} \in
  \pmatspace{\bar{m}}{\bar{m}}\) of degree \(\le \lceil \Delta(P) / m
  \rceil\) and size \(m \le \bar{m} < 3m\), which is such that \(\det(\mat{P}) =
  \det(\mat{\bar{P}})\) and \(\mat{P}^{-1}\) is the principal \(m \times m\)
  submatrix of \(\mat{\bar{P}}^{-1}\).
\end{lemma}

\begin{lemma}[{\cite[Sec.\,9]{Storjohann2003}}]
  \label{lem:lifting}
  Algorithm \textproc{SeriesSol} \cite[Alg.\,4]{Storjohann2003} takes as input
  \(P \in \pmatspace{m}{m}\) of degree \(t\) with \(P(0)\) invertible, \(V \in
  \pmatspace{m}{n}\), and \(s \in \ZZp\), and returns the expansion \((P^{-1}V)
  \rem x^{s t}\) using
  \[
    \bigO{\log(s+1) \left\lceil \frac{s n}{m} \right\rceil \timepmm{m}{t} + \timepmmrec{m}{t}}
  \]
  operations in \(\field\). The term \(\timepmmrec{m}{t}\) comes from a call to
  \cite[Alg.\,1]{Storjohann2003} which is independent of \(V\), namely
  $\textproc{HighOrderComp}(P,\lceil \log_2(s) \rceil-1)$.
\end{lemma}
\begin{proof}
  There is no operation needed when \(t = 0\), i.e., for a constant matrix
  \(P\). For \(t > 0\), this follows from \cite[Prop.\,15]{Storjohann2003} and
  the paragraph that precedes it, applied with \(X = x^{t}\). Indeed this is a
  valid choice since \(X\) has degree \(\deg(P)\) and is coprime with
  \(\det(P)\).  The corresponding algorithm in \cite{Storjohann2003} has two
  requirements on \(s\), which are easily lifted: it should be a power of \(2\)
  (one can compute with the next power of \(2\), i.e.\ \(2^{\lceil \log_2(s)
  \rceil}\), and then truncate at the desired order) and it should be at least
  \(4\) (the case \(s \in \bigO{1}\) is handled in \(\bigO{\lceil n/m \rceil
  \timepmm{m}{t}}\) via a direct Newton iteration).
\end{proof}

\begin{proposition}
  \label{prop:TruncatedInverse}
  Let \(\mat{P} \in \pmatspace{m}{m}\) with \(P(0)\) invertible, and let \(d =
  (d_j)_j \in \NN^m\). Let \(t\) be the degree of the partially linearized
  matrix \(\bar{P}\) as in Lemma~\ref{lem:partial_linearization}, thus with \(t \le
  \lceil \Delta(P) / m \rceil\). Algorithm \texttt{TruncatedInverse} uses
\(\bigO{m^\expmm}\) operations in \(\field\) if \(\deg(P) = 0\), and
  \[
    \bigO{ \timepmmrec{m}{\avgdeg} +
      \left\lceil \frac{\sumTuple{d}}{m\avgdeg} \right\rceil
      \timepmm{m}{\avgdeg}
      \log(m) \log\left(m + \frac{\sumTuple{d}}{t} \right)
    }
  \]
  operations in \(\field\) if \(\deg(P)>0\) (which implies \(t>0\)).  It
  returns \(P^{-1} \rem x^d\), the power series expansion of \(\mat{P}^{-1}\)
  with column \(j\) truncated at order \(d_j\).
  If \(n\) is a parameter such that \(m\) and \(\Delta(P)\) are both in
  \(\bigO{n}\), the above bound is in
  \[
    \bigO{m^\expmm \timepm{\frac{n}{m}}
      \left( \log(n) + \left\lceil \frac{\sumTuple{d}}{n} \right\rceil \log(m) \log(m+\sumTuple{d}) \right)
    }.
  \]
\end{proposition}

\renewcommand{\thealgorithm}{4}

\begin{algorithm}[ht]
  \algoCaptionLabel{TruncatedInverse}{P,d}
  \begin{algorithmic}[1]
    \Require \(\mat{P} \in \pmatspace{m}{m}\) with \(P(0)\) invertible, \(d = (d_1,\ldots,d_m) \in \NN^m\)
    \Ensure the column-truncated inverse \(\mat{P}^{-1} \rem x^d \in \pmatspace{m}{m}\)

    \State \(Q \gets\) zero matrix in \(\pmatspace{m}{m}\) \Comment{stores the result}

    \State \InlineIf{\(d = (0,\ldots,0)\)}{\Return \(Q\)}
        \label{step:TruncatedInverse:zeroorder}

    \State\InlineIf{\(\deg(P) = 0\)}{\Comment{constant matrix inversion, \(\bigO{m^\expmm}\)}}
        \label{step:TruncatedInverse:zerodeg}
        \State \hspace{0.5cm}\(Q \gets P^{-1}\); \(\matcol{Q}{j} \gets 0\) for all \(j\) with \(d_j=0\); \Return \(Q\)

    \State \CommentLine{build partition of \(\{1,\ldots,m\}\) based on truncation order}
        \label{step:TruncatedInverse:partition}
    \State \(\delta \gets \sumTuple{d}/m\); \(\ell \gets \lceil \log_2(m) \rceil\)
    \State \(J_1 \gets \{j \in \range{m} \mid d_j \le 2 \delta\}\); \(n_1 \gets \card{J_1}\)
    \For{\(k \gets 2, \ldots, \ell\)}
      \State \(J_k \gets \{j \in \range{m} \mid 2^{k-1} \delta < d_j \le 2^k \delta\}\);
                  \(n_k \gets \card{J_k}\)
    \EndFor

    \State \CommentLine{partial linearization, note that \(t = \deg(\bar{P}) > 0\)}
    \State \(\mat{\bar{P}} \in \pmatspace{\bar{m}}{\bar{m}} \gets\) matrix built from \(P\) as in Lemma~\ref{lem:partial_linearization}
    \State \(\avgdeg \gets \deg(\mat{\bar{P}})\); \InlineFor{\(k=1,\ldots,\ell\)}{\(s_k \gets \lceil 2^k\delta / \avgdeg \rceil\)}

    \State \CommentLine{store high-order components to avoid redundant iterations}
    \State \(C \gets \textproc{HighOrderComp}[x^t](\bar{P},\lceil \log_2(s_\ell) \rceil - 1)\)
            \Comment{\cite[Alg.\,1]{Storjohann2003}}
            \label{step:TruncatedInverse:HighOrderComp}

    \State \CommentLine{main loop: iteration \(k\) handles \(\matcol{Q}{j}\) for \(j\) in \(J_k\)}
    \For{\(k \gets 1, 2, \ldots, \ell\)}
      \State \(E \in \matspace{\bar{m}}{n_k} \gets \matcols{(\ident{\bar{m}})}{J_k}\) \Comment{columns of \(\bar{m}\times \bar{m}\) identity matrix}

      \State \(F \in \pmatspace{\bar{m}}{n_k} \gets \textproc{SeriesSol}(\bar{P}, E, s_k)\), using the pre-
          \label{step:TruncatedInverse:SeriesSol}
            \StateX{2}computed high-order components \(C\) \Comment{Lemma~\ref{lem:lifting}}

      \State \(\matcols{Q}{J_k} \gets\) first \(m\) rows of \(F\), all truncated at order \((d_j)_{j\in J_k}\)
          \label{step:TruncatedInverse:shaveoff}
    \EndFor
    \State \Return \(Q\)
  \end{algorithmic}
\end{algorithm}

\begin{proof}
  When \(d=(0,\ldots,0)\), the algorithm performs no field operations and
  returns the zero matrix (see Line~\ref{step:TruncatedInverse:zeroorder}). When
  \(d \neq 0\) and \(\deg(P)=0\) (hence \(t=0\)),
  Line~\ref{step:TruncatedInverse:zerodeg} correctly computes \(P^{-1} \rem x^d\)
  in complexity \(\bigO{m^\expmm}\), which is within the claimed cost since
  \(m^\expmm \in \bigO{\timepmm{m}{0}}\). From here on, assume \(d\neq 0\) and
  \(\deg(P)>0\).

  The sets \(J_1,\ldots,J_m\) built at Line~\ref{step:TruncatedInverse:partition}
  are disjoint and, since \(\max_j d_j \le \sumTuple{d} \le 2^\ell \delta\),
  they are such that \(J_1 \cup \cdots \cup J_m = \range{m}\). Note also that
  the cardinality \(n_k = \card{J_k}\) is less than \(m/2^{k-1}\).  We claim
  that at the end of the \(k\)th iteration of the main loop, \(\matcol{Q}{j}\)
  is the column \(j\) of the sought output \(\mat{P}^{-1} \rem x^d\) for all
  \(j \in J_1 \cup \cdots \cup J_k\), which implies the correctness of the
  algorithm. This claim follows from Lemma~\ref{lem:partial_linearization}.  Indeed,
  writing \(t = \deg(\bar{P})\), since the principal \(m \times m\) submatrix
  of \(\mat{\bar{P}}^{-1}\) is \(\mat{P}^{-1}\),
  Line~\ref{step:TruncatedInverse:SeriesSol} computes, in the top \(m\) rows of
  \(F\), all columns \(j \in J_k\) of \(\mat{P}^{-1}\) truncated at order
  \(\lceil 2^k\delta / t \rceil t\), which is at least the target order
  \(d_j\).  The subsequent Line~\ref{step:TruncatedInverse:shaveoff} further
  truncates to shave off the possible extraneous expansion terms, and also
  selects the relevant rows of \(F\). Note that \(t > 0\): if \(\bar{P}\) was
  constant, then the principal \(m \times m\) submatrix of
  \(\mat{\bar{P}}^{-1}\) would be constant, i.e.\ \(P^{-1}\) would be constant,
  which is not the case since \(\deg(P) > 0\).

  As noted in Lemma~\ref{lem:lifting}, the \(k\)th call to \textproc{SeriesSol}
  involves a call to \textproc{HighOrderComp}, which does not depend on the
  matrix \(E\) at this iteration and which will re-compute the same high order
  components as the previous iterations, plus possibly one new such component.
  To avoid this redundancy, we pre-compute all required high-order components
  before the main loop at Line~\ref{step:TruncatedInverse:HighOrderComp}.

  As for complexity, only
  Lines~\ref{step:TruncatedInverse:HighOrderComp} and \ref{step:TruncatedInverse:SeriesSol}
  use arithmetic operations. The construction of \(\bar{P}\) in
  Lemma~\ref{lem:partial_linearization} implies \(\det(\bar{P}(0)) =
  \det(\bar{P})(0) = \det(P)(0) \neq 0\), hence we can apply
  Lemma~\ref{lem:lifting}. Here, \(\bar{P}\) is \(\bar{m} \times \bar{m}\) with \(m
  \le \bar{m} < 3m\), and \(t = \deg(\mat{\bar{P}}) \le \lceil \Delta(P) / m
  \rceil\).

  Hence, using notation \(s_k = \lceil 2^k\delta / \avgdeg \rceil\), the
  complexity is within
  \[
    \bigO{\timepmmrec{m}{\avgdeg} + \sum_{k=1}^{\ell}  \log(s_k+1) \left\lceil \frac{s_k n_k}{m} \right\rceil \timepmm{m}{\avgdeg} },
    \text{ with}
  \]
  \[
    \sum_{k=1}^{\ell}  \log(s_k+1) \left\lceil \frac{s_k n_k}{m} \right\rceil
    \in
    \bigO{\lceil \delta/\avgdeg \rceil \left(\ell^2 + \ell \log(\lceil \delta/\avgdeg \rceil)\right)}
  \]
  thanks to the upper bounds
  \(
    \log_2(s_k+1)
    \le \ell+1 + \log_2(\lceil \delta/\avgdeg \rceil)
  \)
  and \(\lceil s_k n_k/m
  \rceil \le \lceil s_k/2^{k-1} \rceil
  = \lceil 2\delta/\avgdeg \rceil\).

  To obtain the claimed general cost bound, it remains to note that
  \[
    \ell^2 + \ell \log(\lceil \delta/\avgdeg \rceil)
    \in \bigO{\log(m) \log(m\lceil \delta/\avgdeg \rceil)},
  \]
  with \(m\lceil \delta/\avgdeg \rceil \in \bigO{m + \sumTuple{d}/\avgdeg}\).
  For the simplified bound, we first use \(t \ge 1\) to bound \(\log(m +
  \sumTuple{d}/\avgdeg)\) by \(\log(m + \sumTuple{d})\). The assumptions on the
  introduced parameter \(n\) allow us to write \(t \in \bigO{n/m}\). In
  particular, \(\timepmmrec{m}{t}\) is in \(\bigO{m^\expmm \timepm{n/m}
  \log(n/m)}\), which is within the claimed bound.  It remains to observe that
  \begin{align*}
    \left\lceil \frac{\sumTuple{d}}{m\avgdeg} \right\rceil
       & \timepmm{m}{\avgdeg} \in \bigO{ \left( 1 + \frac{\sumTuple{d}}{m\avgdeg} \right) m^\expmm \timepm{t}} \\
    & \subseteq \bigO{ m^\expmm \timepm{\frac{n}{m}} + \frac{\sumTuple{d}}{m} m^\expmm \frac{\timepm{t}}{t}} \\
    & \subseteq \bigO{m^\expmm \timepm{\frac{n}{m}} + \frac{\sumTuple{d}}{n} m^\expmm \timepm{\frac{n}{m}}}
    \subseteq \bigO{\left\lceil \frac{\sumTuple{d}}{n} \right\rceil m^\expmm \timepm{\frac{n}{m}}}.
  \end{align*}
  Here we have used the superlinearity assumption on \(\timepm{\cdot}\), which
  gives us \(\frac{\timepm{t}}{t} \in \bigO{\frac{\timepm{n/m}}{n/m}}\).
\end{proof}

%%%%%%%%%%%%%%%%%%%%%%%%%%%%%%%%%%%%%%%%%%%%%%%%%%%%%%%%%%%%%%%%%%%%%%%%%%%%%%%%%%%
%
%              POLMAT --> TRUNCATED PRODUCT
%
%%%%%%%%%%%%%%%%%%%%%%%%%%%%%%%%%%%%%%%%%%%%%%%%%%%%%%%%%%%%%%%%%%%%%%%%%%%%%%%%%%%

\subsection{Polynomial matrix truncated product}
\label{sec:polmat:TruncatedProduct}

\begin{proposition}
  \label{prop:TruncatedProduct}
  Given \(\mat{F} \in \pmatspace{n}{m}\), \(\mat{G} \in \pmatspace{m}{m}\), and
  \(d = (d_j)_j \in \NN^m\), Algorithm \texttt{TruncatedProduct} uses
  \[
    \bigO{\sum_{0 \le k < \lceil \log_2(m) \rceil} \left\lceil 2^k \frac{n}{m} \right\rceil \timepmm{2^{-k}m}{2^k\left\lceil\frac{D}{m}\right\rceil}}
  \]
  operations in \(\field\) and returns the truncated product \((FG) \rem x^d\)
  that is, \(FG\) with column \(j\) truncated at order \(d_j\). Here, \(D\) is
  the maximum between \(\sumTuple{d}\) and the sum of the degrees of the
  nonzero columns of \(F\).
  If \(m\) is both in \(\bigO{n}\) and \(\bigO{D}\), this cost bound is
  in \(\bigO{m^{\expmm-2} n \timepm{D}}\).
\end{proposition}

\begin{proof}
  For convenience, we denote by \(I_k\) and \(J_k\) the sets \(I\) and \(J\)
  defined at the iteration \(k\) of the main loop of the algorithm. Let also
  \(R^{(k)}\) be the matrix \(R\) at the beginning of iteration \(k\), and
  \(R^{(\ell)}\) be the output \(R\). Let \(F^{(0:k)} = F \rem x^{2^k\delta}\) for
  \(1\le k \le \ell\), with in particular \(F^{(0:\ell)} = F\) since \(\max_j
  d_j \le \sumTuple{d} \le m\delta \le 2^\ell \delta\).

  At the beginning of the first iteration, \(R^{(1)} = (F^{(0)} G) \rem x^{d} =
  (F^{(0:1)} G) \rem x^{d}\). Then, to prove the correctness of the algorithm,
  we let \(k \in \range{\ell}\) and assume \(R^{(k)} = (F^{(0:k)} G) \rem
  x^{d}\), and we show \(R^{(k+1)} = (F^{(0:k+1)}G) \rem x^{d}\).

  First, consider \(j\not\in J_k\). The column \(j\) of \(R^{(k)}\) is not
  modified by iteration \(k\), i.e.\ \(\matcol{R}{j}^{(k)} =
  \matcol{R}{j}^{(k+1)}\). On the other hand, one has \(d_j < 2^k\delta\),
  hence \(F^{(0:k)} \rem x^{d_j} = F^{(0:k+1)} \rem x^{d_j}\). We obtain
  \[
    \matcol{R}{j}^{(k+1)}
    = \matcol{R}{j}^{(k)}
    = (F^{(0:k)}\matcol{G}{j}) \rem x^{d_j}
    = (F^{(0:k+1)}\matcol{G}{j}) \rem x^{d_j}.
  \]
  meaning that the sought equality holds for the columns \(j \not \in J_k\).

\renewcommand{\thealgorithm}{5}

\begin{algorithm}[t]
  \algoCaptionLabel{TruncatedProduct}{F,G,d}
  \begin{algorithmic}[1]
    \Require \(\mat{F} \in \pmatspace{n}{m}\), \(\mat{G} \in \pmatspace{m}{m}\),
        \(d = (d_1,\ldots,d_m) \in \NN^m\)
    \Ensure the column-truncated product \((FG) \rem x^d \in \pmatspace{n}{m}\)

    \State
    \(\gamma \gets\) sum of the degrees of the nonzero columns of \(F\)

    \State
    \(D \gets \max(\sumTuple{d},\gamma)\); ~~
    \(\delta \gets \lceil D/m\rceil\); ~~
    \(\ell \gets \lceil\log_2(m)\rceil\)

    \State
    write \(\mat{F} = \mat{F}^{(0)} + \sum_{1 \le k < \ell} \mat{F}^{(k)} x^{2^k\delta}\)
    with each \(F^{(k)}\) in \(\pmatspace{n}{m}\),
    \(\deg(F^{(k)}) < 2^k\delta\) for \(1 \le k < \ell\), and \(\deg(F^{(0)}) < 2\delta\)

    \State
    \(R \in \pmatspace{n}{m} \gets (F^{(0)}G) \rem x^d\) \Comment{stores the result}
    \label{step:TruncatedProduct:first_product}

    \For{\(k \gets 1,\ldots, \ell-1\)}
    \State \(I \gets \{i\in\range{m} \mid \cdeg{\matcol{F}{i}} \ge 2^k\delta\}\)
    \State \(J \gets \{j\in\range{m} \mid d_j \ge 2^k\delta\}\)
    \State \(e \gets (d_j - 2^k\delta)_{j \in J}\)
    \State \(\matcols{R}{J} \gets \matcols{R}{J} + \left((\matcols{F}{I}^{(k)}\submat{G}{I}{J}) \rem x^{e}\right) x^{2^k\delta} \)
        \label{step:TruncatedProduct:partial_product}
    \EndFor
    \State \Return \(R\)
  \end{algorithmic}
\end{algorithm}
  Now, consider \(j \in J_k\). For \(i\not\in I_k\), one has
  \(\cdeg{\matcol{F}{i}} < 2^k\delta\), hence \(\matcol{F}{i}^{(k)} = 0\): it
  follows that \(F^{(k)} \matcol{G}{j} =
  \matcols{F}{I}^{(k)}\submat{G}{I}{j}\). Thus
  Line~\ref{step:TruncatedProduct:partial_product} computes
  \begin{align*}
    \matcol{R}{j}^{(k+1)} \gets ~
       & \matcol{R}{j}^{(k)} + \left((\matcols{F}{I}^{(k)}\submat{G}{I}{j}) \rem x^{d_j - 2^k\delta}\right) x^{2^k\delta} \\
       & = (F^{(0:k)} \matcol{G}{j}) \rem x^{d_j} + (x^{2^k\delta}F^{(k)}\matcol{G}{j}) \rem x^{d_j} \\
       & = (F^{(0:k+1)}\matcol{G}{j}) \rem x^{d_j}.
  \end{align*}
  This completes the proof of correctness.

  For complexity, note that the definition of \(D\) gives \(\card{I_k} \le
  2^{-k}m\) and \(\card{J_k} \le 2^{-k}m\); also, only
  Lines~\ref{step:TruncatedProduct:first_product} and \ref{step:TruncatedProduct:partial_product}
  use arithmetic operations.

  At Line~\ref{step:TruncatedProduct:first_product}, we first compute \(F^{(0)}G\),
  then truncate. The left and right matrices in this product are respectively
  \(n\times m\) of degree \(<2\delta\), and \(m\times m\) with sum of column
  degrees \(\le m\delta\). The product can be performed by expanding the
  columns of \(G\) into \(\le 2m\) columns all of degree \(\le\delta\), which leads
  to the complexity \(\bigO{\lceil \frac{n}{m} \rceil \timepmm{m}{\delta}}\).

  At Line~\ref{step:TruncatedProduct:partial_product}, the multiplication by a
  power of \(x\) is free. The sum consists in adding two matrices with \(n\)
  rows and with column degrees strictly less than \((d_j)_{j\in J_k}\)
  entry-wise.
  This costs \(\bigO{n \sumTuple{d}}\) operations in \(\field\) at each
  iteration, hence \(\bigO{\ell n \sumTuple{d}}\) in total. Using the trivial
  lower bound on \(\timepmm{\cdot}{\cdot}\) shows that this is within the
  claimed overall cost bound:
  \begin{align*}
    \ell n \sumTuple{d} \le \sum_{k=0}^{\ell-1} n D
                        & \le \sum_{k=0}^{\ell-1} \left\lceil 2^k \frac{n}{m} \right\rceil (2^{-k}m)^2 2^k \left\lceil \frac{D}{m} \right\rceil \\
                        & \in \bigO{\sum_{k=0}^{\ell-1} \left\lceil 2^k \frac{n}{m} \right\rceil \timepmm{2^{-k}m}{2^k\left\lceil\frac{D}{m}\right\rceil}}.
  \end{align*}

  Finally, for the truncated product, we first multiply
  \(\matcols{F}{I_k}^{(k)}\submat{G}{I_k}{J_k}\) and then truncate. The left
  matrix in this product has \(n\) rows, \(\le 2^{-k}m\) columns, and degree
  \(<2^k\delta\). The right matrix has row and columns dimensions both \(\le
  2^{-k}m\), and sum of column degrees \(\le\sumTuple{d}\). The product can be
  performed by expanding the columns of \(\submat{G}{I_k}{J_k}\) into \(\le
  2^{1-k}m\) columns all of degree \(\le \sumTuple{d}/(2^{-k}m) \le
  2^k\delta\), which leads to the complexity \(\bigO{\lceil 2^k \frac{n}{m}
  \rceil \timepmm{2^{-k}m}{2^k\delta}}\). Summing the latter bound for \(k \in
  \{1,\ldots,\ell-1\}\), and adding the term \(k=0\) for
  Line~\ref{step:TruncatedProduct:first_product}, yields the claimed cost bound.

  The final simplified complexity bound follows from the assumptions mentioned
  in Section~\ref{sec:polmat:timefunc}: \(\timepmm{\mu}{\delta}\) is in
  \(\bigO{\mu^\expmm \timepm{\delta}}\), \(\timepm{\cdot}\) is superlinear, and
  \(\expmm>2\) implies that the sum \(\sum_{0 \le k < \ell}
  2^{k(2-\expmm)}\) is bounded by a constant.
\end{proof}

%%%%%%%%%%%%%%%%%%%%%%%%%%%%%%%%%%%%%%%%%%%%%%%%%%%%%%%%%%%%%%%%%%%%%%%%%%%%%%%%%%%%%

\clearpage 


\newcommand{\Gathen}{\relax}

\appendix

%%%%%%%%%%%%%%%%%%%%%%%%%%%%%%%%%%%%%%%%%%%%%%%%%%%%%%%%%%%%%%%%%%%%%%%%%%%%%%%%%%%
%
%              KRYLOV STUFF
%
%%%%%%%%%%%%%%%%%%%%%%%%%%%%%%%%%%%%%%%%%%%%%%%%%%%%%%%%%%%%%%%%%%%%%%%%%%%%%%%%%%%

\section{Decomposition of the space \texorpdfstring{$\field^n$}{F**n}}

\label{sec:krylovstuff}

The following is part of the basic material when studying the behavior of a
linear operator $A$ and the decomposition of $\field ^n$ into cyclic
subspaces~\cite[Chap.\,VII]{Gantmacher1960a}, which corresponds to the diagonal
matrix form of Smith and the block diagonal form of Frobenius. The same
concepts can also be used for decompositions associated with the triangular
form of Hermite~[\citealp{Vil97};\citealp[Chap.\,9]{Sto00}], or more general
forms such as column reduced ones~\cite[Sec.\,6.4.6, p.\,424]{Kailath80}.

\myparagraph{Existence of a Krylov basis}
Let \(A \in \matspace{n}{n}\) and \(U \in \matspace{n}{m}\), and 
for $1\leq j\leq m$, let $d_j$ be the first integer such that
\begin{equation} \label{eq:characd_app} 
A^{d_j}u_j \in \operatorname{Span} (u_j, Au_j, \ldots, A^{d_j-1} u_j) +\orb{A}{U_{*,1..j-1}}.
\end{equation}
We prove that the columns of \(\kry{A}{U}{d}\) form a basis of \(\orb{A}{U}\).

Given a subspace $\mathcal E \subseteq \field ^n $ invariant with respect to
$A$, we say that two vectors $u,v\in \field^n $ are congruent modulo $\mathcal
E$ if and only if $u-v \in \mathcal E$, and we write $u\equiv v \bmod \mathcal
E$. For a fixed $u$, the set of polynomials $p\in \field[x]$ such that $p(A)u
\equiv 0 \bmod \mathcal E$ is an ideal of $\field[x]$, generated by a monic
polynomial which is the minimal polynomial of $u$ modulo $\mathcal E$. In
particular, if $p(A)u \equiv 0 \bmod \mathcal E$, then $q(A)u \equiv 0 \bmod
\mathcal E$ for all multiples $q\in \field[x]$ of $p$.

The proof of the characterization in \cref{eq:characd_app} is by induction on
$m$. For one vector, $d_1$ is the first integer such that 
\[
  (x^{d_1})(A)=A^{d_1}u_1 \equiv 0 \bmod \operatorname{Span} (u_1, Au_1, \ldots, A^{d_1-1} u_1).
\]
Therefore all the subsequent vectors  $(x^{d_1}x^k)(A)=A^{d_1+k}u_1$ with
$k\geq 0$ are zero modulo \(\kry{A}{u_1}{d_1}\) which thus forms a basis of
\(\orb{A}{u_1}\). Then assume that the property holds for all \(U\) of column
dimension \(m\ge 1\): \(\kry{A}{U}{d}\) is a basis of \(\orb{A}{U}\). For
$v\in \field ^m$, let~$d_{m+1}$ be the smallest integer so that $A^{d_{m+1}}v$
is a combination of the previous iterates of $v$ and the vectors in
\(\orb{A}{U}\). By analogy to the case $m=1$, all subsequent vectors
\(A^{d_{m+1}+k} v\) are also linear combinations of the vectors in
$\operatorname{Span} (v, Av, \ldots, A^{d_{m+1}-1} v) + \orb{A}{U}.$ So the
latter subspace is $\orb{A}{[U,v]}$, and by induction hypothesis the columns of
\(\kry{A}{v}{d_{m+1}}\) and \(\kry{A}{U}{d}\) form one of its bases. 

\myparagraph{Maximal Krylov indices} The tuple $d$ constructed this way is
lexicographically maximal so that \(\kry{A}{U}{d}\) is a basis of $\orb{A}{U}$.
The existence of an $\ell$ such that $(d'_1, \ldots, d'_\ell, \ldots d'_m)$ is
another suitable tuple with $d'_\ell<d_\ell$ would indeed contradict the fact
that $d_\ell$ corresponds to the smallest linear dependence.  

\section{A slightly different algorithm}
\label{sec:alternativeto}

The matrices $xI-A$ and~$I-xA$ considered in Algorithm \texttt{MaxIndices} and
Algorithm \texttt{KrylovMatrix} mirror each other. As we explain below, $xI-A$
could also be used to compute a Krylov matrix. (Algorithm \texttt{KrylovMatrix}
instead considers $I-xA$ simply for a slightly more convenient presentation.)

Consider a minimal kernel basis $\trsp{[\trsp{S} \;\; \trsp{T}]}$ of $[xI-A
\;\; -U]$ as in Algorithm \texttt{MaxIndices}. Since $(xI-A)^{-1}U$ is strictly
proper and $S=(xI-A)^{-1}UT$, the column degrees in $T$ are greater than those
in $S$, and~$T$ is column reduced  since the kernel basis is
(Lemma~\ref{lem:kernel_basis}). Let $d_j$ be the degree of the $j$th column of $T$,
and let \(d = (d_1,\ldots,d_m)\). We denote by \(x^d\) the \(m\times m\)
diagonal polynomial matrix with diagonal entries $(x^{d_1}, \ldots, x^{d_m})$.
Substituting $1/x$ for $x$ in $(xI-A)S-UT=0$, and multiplying on the right by
$x^d$ we get 
\[
  (I-xA) S(1/x) x^{d-1} - U T(1/x) x^d = 0.
\]
By definition of the $d_j$'s, the right-hand term is a polynomial matrix, say
$U\hat T$. Since the $j$th column of $S$ has degree at most \(d_j - 1\) (and is
zero if $d_j=0$), the left-hand term is also a polynomial matrix, say
$(I-xA)\hat S$. 
It follows that the columns of $\trsp{[\trsp{\hat{S}} \;\; \trsp{\hat{T}}]}$
are in the kernel of $[I-xA \;\; -U]$, and we can now also verify that they
form a basis. Since $\trsp{[\trsp{S} \;\; \trsp{T}]}$ is itself a minimal
basis, it is irreducible, i.e.\ has full rank for all finite values
of~$x$~\cite[Thm\,6.5-10, p.\,458]{Kailath80}. Equivalently, $S$ and $T$ are
coprime~\cite[Lem.\,6.3-6, p.\,379]{Kailath80}, so there exists $V$ unimodular
such that 
$V \trsp{[\trsp{S} \;\; \trsp{T}]} = \trsp{[I \;\; 0]}$. It follows that 
\[
  V(1/x)
  \begin{bmatrix} x {\hat S} \\ {\hat T} \end{bmatrix}
  = \begin{bmatrix} {I} \\ {0} \end{bmatrix} x^d .
\]
and since $V$ is unimodular, the rank of $\trsp{[\trsp{\hat{S}(x_0)} \;\;
\trsp{\hat{T}(x_0)}]}$ is full for all values $x_0\neq 0$ of $x$. The fact that
$T$ is column reduced adds that $\hat T(0)$ is invertible, which means that
$\trsp{[\trsp{\hat{S}} \;\; \trsp{\hat{T}}]}$ is irreducible.
By irreducibility (\cite[Thm\,6.5-10, p.\,458]{Kailath80} again, here after
column reduction), we get as announced that $\trsp{[\trsp{\hat{S}} \;\;
\trsp{\hat{T}}]}$ is a kernel basis of $[I-xA \;\; -U]$.

If $A$ is nonsingular, then, in a similar way to what we saw with Lemma~\ref{lem:kryloveval:constant_invertible}, 
we can prove that $T(0)$ is invertible. In this case it follows that $\hat T$ is column reduced, and that 
$\trsp{[\trsp{\hat{S}} \;\;
\trsp{\hat{T}}]}$ is a minimal kernel basis of $[I-xA \;\; -U]$. Minimality is not guaranteed if $A$ is singular.  

We see that Algorithm \texttt{KrylovMatrix} could be modified using the same
kernel basis as in Algorithm \texttt{MaxIndices}, and by inserting an instruction
to work with~$\hat T$ afterwards. 
Since $\hat T(0)$ si invertible ($T$ is column reduced), 
the fact that this basis may not be minimal is not a problem for the power series expansion and the subsequent steps in Algorithm \texttt{KrylovMatrix}. Also, since $\sumTuple{\cdeg{\hat T}} \le \sumTuple{\cdeg{T}} \le n$, the cost bound in Lemma~\ref{lem:KrylovMatrix} is not affected. 
In particular, when using Algorithm \texttt{MaxIndices} and Algorithm \texttt{KrylovMatrix} in 
succession as in Algorithm \texttt{MaxKrylovBasis}, this would lead to computing
only one kernel basis instead of two.

\clearpage 

\end{document}